\documentclass{article}
\usepackage{amssymb}
\usepackage{amsmath}
\usepackage{amsthm}
\usepackage{amsfonts}
\usepackage{latexsym}
\usepackage{textcomp}
\usepackage{graphicx}
\usepackage{url}
\newcommand{\gr}[1]{\boldsymbol{#1}}
\renewcommand{\th}{\vartheta}
\newcommand{\R}{\mathbb{R}}
\newcommand{\de}[2]{\frac{\partial #1}{\partial #2}}
\newcommand{\weak}{\rightharpoonup}
\newtheorem{prop}{Proposition}
\newtheorem{defn}{Definition} 
\newtheorem{thm}{Theorem}
\newtheorem{rem}{Remark}

\renewcommand{\div}{\operatorname{div}}
\newcommand{\eps}{\varepsilon}
\renewcommand{\th}{\vartheta}
\newcommand{\vett}[1]{\boldsymbol{#1}}
\usepackage[top=2.5cm,right=2.3cm,left=2.3cm, bottom=2.5cm]{geometry}
\begin{document}
	\title{Thermal convection in a higher velocity gradient\\and higher temperature gradient fluid}
	\author{G. Giantesio$^{(1)}$\thanks{\texttt{giulia.giantesio@unicatt.it}}, A. Girelli$^{(1)}$\thanks{\texttt{alberto.girelli@unicatt.it}}, C. Lonati$^{(2)}$\thanks{\texttt{chiara.lonati01@universitadipavia.it}},\\
		A. Marzocchi$^{(1)}$\thanks{
\texttt{alfredo.marzocchi@unicatt.it}}, A. Musesti$^{(1)}$\thanks{
\texttt{alessandro.musesti@unicatt.it}}, B. Straughan$^{(3)}$\thanks{\texttt{brian.straughan@durham.ac.uk}}\\
		\bigskip\\
		\normalsize $^{(1)}$ Dipartimento di Matematica e Fisica ``N. Tartaglia", 	\normalsize Università Cattolica del Sacro Cuore,\\
		\normalsize via della Garzetta 48, I-25133 Brescia, Italy\\
		\normalsize$^{(2)}$ Dipartimento di Matematica ``F. Casorati'',
		\normalsize Università di Pavia,\\
		\normalsize via Ferrata 5, I-27100 Pavia, Italy.\\
		\normalsize$^{(3)}$ Department of Mathematical Sciences,
		\normalsize	University of Durham,\\ 
		\normalsize	Stockton Road, Durham DH1 3LE, United Kingdom. \\
		}
	\date{}
	\maketitle
	\begin{abstract}
		We analyse a model for thermal convection in a class of generalized Navier-Stokes equations containing fourth order spatial derivatives
		of the velocity and of the temperature. The work generalises the isothermal model of A. Musesti.
		We derive critical Rayleigh and wavenumbers for the onset of convective fluid motion paying careful attention to the variation
		of coefficients of the highest derivatives.
		In addition to linear instability theory we include an analysis of fully nonlinear stability theory.
		The theory analysed possesses a bi-Laplacian term for the velocity field and also for the temperature field.
		It was pointed out by
		E. Fried and M. Gurtin that higher order terms represent micro-length effects and these phenomena are very important in flows
		in microfluidic situations. 
		We introduce temperature into the theory via a Boussinesq approximation where the density of the body force term is allowed
		to depend upon temperature to account for buoyancy effects which arise due to expansion of the fluid when this is heated.
		We analyse a meaningful set of boundary
		conditions which are introduced by Fried and Gurtin as conditions of strong adherence, and these are
		crucial to understand the effect of the higher order derivatives upon convective motion in a microfluidic 
		scenario where micro-length effects are paramount.
		The basic steady state is the one of zero velocity, but in contrast to the classical theory the temperature field is nonlinear in the 
		vertical coordinate. This requires care especially dealing with nonlinear theory and also leads to some novel effects.
	\end{abstract}
	\section{Introduction}
\label{S:Intro}
	There is growing interest in the fluid dynamics literature in theories which are generalizations of the 
	Navier-Stokes equations, cf. \cite{Bresch:2007}, \cite{Bresch:2022}, \cite{Damazio:2016}, \cite{Degiovanni:2020},
	\cite{FriedGurtin:2006}, \cite{Giusteri:2011}, \cite{GoudonVasseur:2016}, \cite{Guillen:2007}, \cite{JabourBondi:2022},
	\cite{KalantarovTiti:2018}, \cite{Musesti:2009}, \cite{SlomkaDunkel:2015}, 
\cite{Straughan:2021AMO,Straughan:2021AMO2,Straughan:2023AMO,Straughan:2023EPJP},
	\cite{Wang:2017}, \cite{Zvyagin:2013}. 
	Much of this interest is driven by applications in the microfluidics industry where flows are in very small tubes and channels,
	see e.g. \cite{Christov:2018jfm}, \cite{WangChristov:2019prsa}, \cite{Wang:2022mrc}.
	\cite{FriedGurtin:2006} argue that when flow dimensions are small then length scale effects become dominant and the stress
	tensor should depend not only on the velocity gradient, but also on higher gradients of velocity. This led 
	\cite{FriedGurtin:2006} to produce a generalized Navier-Stokes theory where the momentum equation contains in addition
	to the Laplacian of the velocity field, a term with the bi-Laplacian of the velocity. The theory of \cite{FriedGurtin:2006}
	was completed by \cite{Musesti:2009} who gave the full form of constitutive theory for the stress tensor.
	
	Other theories for incompressible fluids which involve a bi-Laplacian are reviewed by \cite{Straughan:2023EPJP} who
	discusses the couple stress theory of \cite{Stokes:1966} and the dipolar fluid theory of \cite{BleusteinGreen:1967}. The latter
	theory is believed appropriate to the case where the fluid contains long molecules and \cite{BleusteinGreen:1967} expand the velocity
	field 
	$v_i({\bf x},t)$
	in a Taylor series 
	\begin{equation*}
		v_i({\bf y},t)=v_i({\bf x},t)+v_{i,j}({\bf x},t)(y_j-x_j)+\dots\,,
	\end{equation*}
	to explain the inclusion of
	$v_{i,j}$
	and
	$v_{i,jk}$
	in the constitutive theory for a dipolar fluid. \cite{Straughan:2023EPJP} develops a theory for thermal convection in the
	Fried-Gurtin-Musesti framework where the momentum equation contains the bi-Laplacian of
	$v_i$.
	
	Within the field of Solid Mechanics higher gradient theories are well established, see e.g. \cite{Fabrizio:2022mrc},
	\cite{Fabrizio:2022jts}, \cite{Iesan:2023}, and the many references therein. These authors give convincing arguments to include not only higher
	derivatives of displacement or velocity, but when temperature effects are present, they argue for the inclusion of higher
	derivatives of temperature in the constitutive theory. This is closely related to the phenomenon of microtemperatures
	which is prevalent in the Continuum Mechanics literature, see e.g. \cite{Aouadi:2020}, \cite{Bazarra:2021}, and the many
	references therein. In this case one surrounds a point 
	${\bf x}$
	by a microelement of diameter
	$d$
	and one writes the temperature
	$T({\bf x},t)$
	in the form, see e.g. \cite{Bazarra:2021},
	\begin{equation*}
		T({\bf y},t)=T({\bf x},t)+T_j({\bf x},t)(y_j-x_j)+O(d^2)\,,
	\end{equation*}
	where
	$T_j$
	are quantities known as microtemperatures which represent the variation of the temperature inside the microelement.
	
	In this article we specialize this concept and regard the expansion of
	$T$
	as a Taylor series to find
	\begin{equation*}
		T({\bf y},t)=T({\bf x},t)+T_{,j}({\bf x},t)(y_j-x_j)+\dots\,.
	\end{equation*}
	We argue that in microfluidic situations not only are higher gradients of velocity important, but also higher gradients
	of temperature should be taken into account. We essentially employ a Fried-Gurtin-Musesti theory but we allow the heat
	flux,
	$q_i$,
	to depend on
	$T_{,m},T_{,mn}$
	and
	$T_{,mnp}$.
	In order to have a heat flux linear in these variables in an isotropic fluid we then have
	\begin{equation}\label{E:qEq}
		q_i=-k_1T_{,i}+k_2\Delta T_{,i}\,,
	\end{equation}
	where 
	$\Delta$
	is the three-dimensional Laplacian and 
	$k_1,k_2$
	are positive constants.
	Such expressions are already employed in the Solid Mechanics case, see \cite{Fabrizio:2022jts}, although there the coldness function $1/T$ is utilized, and see \cite{Iesan:2023}. In an independent approach \cite{Christov:2007} has argued that for heat conduction micro
effects will necessitate an equation like \eqref{E:qEq} for a complete description of the temperature field, 
especially due to relations at the microscopic level. 
This argument has been substantiated using homogenization theory by \cite{NikaMuntean:2022} and \cite{Nika:2023}. 
	
	We study thermal convection in a Fried-Gurtin-Musesti incompressible fluid, allowing also for higher 
	temperature gradients as in \eqref{E:qEq} and analyzing in detail the setting  where a horizontal layer of liquid is heated from below. The main results are existence of a solution and the derivation of precise conditions, in linear and nonlinear stability, under which convective motion is possible. The results differ
	from the classical theory of thermal convection in a Navier-Stokes fluid not only due to the bi-Laplacian term involving the 
	velocity field, but also because the basic steady state solution is nonlinear in the vertical coordinate,
	$z$,
	as opposed to the classical situation where the steady state is linear in
	$z$.
	
	Stability studies in the classical theory of fluid flow and thermal convection are still continuing to be highly relevant in modern fluid
	dynamic research, see e.g.
	\cite{Bissell:2016},
	\cite{Capone:2022PV},
	\cite{Eltayeb:2017},
	\cite{Hughes:2021},
	\cite{Samanta:2020},
	\cite{WangChen:2022pf,WangChen:2022}. Due to the application of this work in microfluidic situations, and in cases where the molecular structure of the fluid
	contains long molecules, or where additives affect the fluid behaviour such as in solar pond technology in the renewable
	energy sector, we believe this work will be very useful.\\
	The plan of the paper is the following. In Section 2 we present the fundamental equations of the problem. In Section 3 we carry out an existence and uniqueness result. We then specify the problem to the mentioned setting in Section 4, and in Section 5 and 6 we develop a careful study for linear and nonlinear stability. Numerical results are presented in Section 7.
	
	\section{Generalized Navier-Stokes model for thermal convection} 
	\label{S:Model}
	
	If we employ 
	a Boussinesq approximation, see \cite{Barletta:2022}, where the density is a constant,
	$\rho_0$,
	apart from in the buoyancy term in the body force, then the momentum equation arising from the Fried-Gurtin-Musesti theory
	has form
	\begin{equation}\label{E:Mom}
		v_{i,t}+v_jv_{i,j}=-\frac{1}{\rho_0}p_{,i}+\nu\Delta v_i-{\hat\xi}\Delta^2 v_i-\alpha g_iT\,,
	\end{equation}
	where
	$p({\bf x},t)$
	is the pressure,
	$g_i$
	is the gravity vector, 
	$\nu$
	is the kinematic viscosity,
	${\hat\xi}$
	is a hyperviscosity coefficient,
	$\Delta$
	is the Laplacian in 3 dimensions, and
	$\alpha$
	is the thermal expansion coefficient of the fluid which arises through the density representation in the body force term, namely
	\begin{equation*}
		\rho=\rho_0(1-\alpha(T-T_0)).
	\end{equation*}
	In \eqref{E:Mom} and throughout
	we employ standard indicial notation together with the Einstein summation convention. For example,
	the divergence of the velocity field is  
$$v_{i,i}\equiv\sum_{i=1}^3v_{i,i}=
		\frac{\partial v_1}{\partial x_1}+\frac{\partial v_2}{\partial x_2}+\frac{\partial v_3}{\partial x_3}=\frac{\partial u}{\partial x}+\frac{\partial v}{\partial y}+\frac{\partial w}{\partial z},
	$$
	where ${\bf v}=(v_1,v_2,v_3)\equiv(u,v,w)$ and ${\bf x}=(x_1,x_2,x_3)\equiv(x,y,z)$. A further example is
	\begin{equation*}
		v_iT_{,i}\equiv\sum^3_{i=1}v_iT_{,i}=
		u\frac{\partial T}{\partial x}
		+v\frac{\partial T}{\partial y}
		+w\frac{\partial T}{\partial z},
	\end{equation*}
	for a function 
	$T$
	depending upon 
	${\bf x}$, $t$.
	
	Since the fluid is
	incompressible, the velocity field satisfies 
	\begin{equation}\label{E:Cty}
		v_{i,i}=0.
	\end{equation}
	The equation of balance of energy employing a Boussinesq approximation \cite{Barletta:2022}, together with equation
	\eqref{E:qEq} becomes
	\begin{equation}\label{E:Teq}
		T_{,t}+v_iT_{,i}=\kappa_1\Delta T-\kappa_2 \Delta^2T.
	\end{equation}
	The coefficients 
	$\kappa_1$
	and
	$\kappa_2$
	represent
	$k_1$
	and
	$k_2$
	divided by
	$\rho_0c_p$
	where
	$c_p$
	is the specific heat at constant pressure of the fluid.
 
	Thus, our model for non-isothermal fluid movement consists of equations \eqref{E:Mom}, \eqref{E:Cty} and \eqref{E:Teq}.

	\section{Existence theory} 
	\label{S:Existence}	
	For $d>0,\lambda=\sqrt{\kappa_1/\kappa_2}\in\R$, we define the function  
	$$h_{d,\lambda}(x):=
	\begin{cases}
		-\dfrac{\lambda x\cosh(\lambda d)-\sinh(\lambda x)}{\lambda d\cosh(\lambda d)-\sinh(\lambda d)}& \lambda\neq0\\
		\noalign{\smallskip}
		\dfrac{x^3-3d^2x}{2d^3}&\lambda=0.
	\end{cases}
	$$
	The function $h_{d,\lambda}$ is readily seen to be $C^\infty$, bounded on $[-d,d]$ for every $\lambda\in\R$, such that $h_{d,\lambda}'(\pm d)=0$ and $\lambda^2h_{d,\lambda}''-h_{d,\lambda}''''=0$ for every $d>0,\lambda\in\R$. 
	
	It is also immediate to verify that
	\begin{equation}
 \label{tbar}\bar T(z)=\frac{T_L-T_U}{2}h_{d/2,\lambda}\left(z-\frac{d}{2}\right)+\frac{T_L+T_U}{2}
 \end{equation}
	verifies
	$$
	\kappa_1\frac{d^2{\bar T}}{dz^2}-\kappa_2\frac{d^4{\bar T}}{dz^4}=0,\quad \bar T(0)=T_L,\ \bar T(d)=T_U,\ \bar T'(0)=\bar T'(d)=0$$
	with $\kappa_1\geq 0,\kappa_2>0$, $T_L>T_U$. 
	The function $\bar T$ is easily seen to be of class $C^\infty$ on $[0,d]$ and uniformly bounded on $[T_U,T_L]$.
	
	Set $\th=T-\bar T$ and $\beta=(T_L-T_U)/d$. The equations \eqref{E:Mom}, \eqref{E:Cty} and \eqref{E:Teq} become:
	\begin{equation}
		\begin{aligned}
			\de{{\bf v}}{t}+({\bf v}\cdot\nabla){\bf v}=&-\frac{1}{\rho_0}\nabla p+\nu\Delta {\bf v}-\hat{\xi}\Delta^2{\bf v}-\alpha\vett g(\th+\bar T)\\
			\de{\th}{t}+{\bf v}\cdot\nabla\th= & \frac{\beta d}{2}w\dfrac{\lambda \cosh(\lambda \frac{d}{2})-\lambda\cosh(\lambda (z-\frac{d}{2}))}{\lambda \frac{d}{2}\cosh(\lambda \frac{d}{2})-\sinh(\lambda \frac{d}{2})}+\kappa_1\Delta\th-\kappa_2\Delta^2\th\\
   &+\kappa_1
   \frac{\beta d}{2}\dfrac{\lambda^2\sinh(\lambda (z-\frac{d}{2}))}{\lambda \frac{d}{2}\cosh(\lambda \frac{d}{2})-\sinh(\lambda \frac{d}{2})}-\kappa_2\frac{\beta d}{2}\dfrac{\lambda^4\sinh(\lambda (z-\frac{d}{2}))}{\lambda \frac{d}{2}\cosh(\lambda \frac{d}{2})-\sinh(\lambda \frac{d}{2})}.
		\end{aligned}
		\label{eq:sist}
	\end{equation}
Let $D$ be a domain in $\R^2$ and $\Omega:=D\times]0,d[$. System \eqref{eq:sist} is going to hold in $\Omega$. As for the boundary conditions, we will first suppose that
	\begin{equation}
		\label{ey conditions, q:condcont}
		\begin{aligned}
			&{\bf v}(x,y,0)={\bf v}(x,y,d)=\de{{\bf v}}{{\bf n}}(x,y,0)=\de{{\bf v}}{{\bf n}}(x,y,d)={\bf 0},\\
			&\th(x,y,0)=\th(x,y,d)=\de{\th}{{\bf n}}(x,y,0)=\de{\th}{{\bf n}}(x,y,d)=0.	
		\end{aligned}
	\end{equation}
On $\partial D$, which will be the boundary $\mathcal H^1$-a.e. regular of a bounded domain $D$ tiling the plane, we will suppose homogeneous boundary conditions for ${\bf v}$ and periodic boundary conditions for $\th$, for a.e. $z\in ]0,d[$:
\begin{equation}\label{eq:ondeD}
    {\bf v}_{|\partial D}=\de{{\bf {\bf v}}}{{\bf n}}_{|\partial D}={\bf 0},\qquad \th\hbox{ and }
    \de{\th}{{\bf n}}\hbox{ periodic on $\partial D$ for a.e. $z\in ]0,d[$}.
\end{equation}

We remark that more general boundary conditions may be considered: the following existence and uniqueness results hold, indeed, for every solution $({\bf v},\th )$ such that their boundary integrals and the boundary integrals of their normal derivatives vanish. We restrict ourselves here to a paradigmatic case for the sake of simplicity.
	Equations \eqref{eq:sist} fit into the following class:
	\begin{equation}
		\begin{aligned}
			&\de{{\bf v}}{t}+({\bf v}\cdot\nabla){\bf v}=-\frac{1}{\rho_0}\nabla p+\nu\Delta {\bf v}-\hat{\xi}\Delta^2 {\bf v}+{\bf L}_1(z)\th+{\bf G}_1(z)\\
			&\de{\th}{t}+{\bf v}\cdot\nabla\th=\kappa_1\Delta\th-\kappa_2\Delta^2\th+{\bf L}_2(z)\cdot {\bf v}+G_2(z)
		\end{aligned}
		\label{eq:semiastr}
	\end{equation}
	where ${\bf L}_1$ and ${\bf L}_2$ are uniformly bounded linear operators and ${\bf G}_1$ and $G_2$ are bounded regular functions of $z$. In particular, there exists $C_d\geq0$ such that
	$$||{\bf G}_1||_{L^\infty(0,d)}\leq C_d,\quad ||G_2||_{L^\infty(0,d)}\leq C_d$$
	where the constant depends only on the thickness $d$. 
 
	We will treat existence theory for a system like \eqref{eq:semiastr} with boundary conditions \eqref{ey conditions, q:condcont}--\eqref{eq:ondeD}. Let us first introduce the appropriate functional setting.
	
	\subsection{Functional spaces, general estimates and functional setting}
	
	Let $\Omega$ be an open domain in $\R^3$ with regular boundary. We introduce the following functional spaces:
	$$
	\begin{aligned}
		C_{\rm div}(\Omega)&=\{{\bf v}\in C_c^\infty(\Omega):\div {\bf v}=0\}\\
		J^2(\Omega)&=\hbox{the closure of $C_{\rm div}(\Omega)$ in $L^2(\Omega)$}\\
		G^2(\Omega)&=\{{\bf f}\in L^2(\Omega):\int_\Omega {\bf f}\cdot {\bf v}\,dx=0 \hbox{ for all } {\bf v}\in C_{\rm div}(\Omega)\}\\
		J^{m,2}(\Omega)&=\hbox{the closure of $C_{\rm div}(\Omega)$ in $H^m_0(\Omega)$}\\
        H^{m,2}_{per}(D)&=\hbox{the subspace of periodic functions in $H^m(D)$}\\
        G^{-m,2}(\Omega)&=\{{\bf f}\in H^{-m}(\Omega):\langle {\bf f},{\bf v}\rangle=0 \hbox{ for all } {\bf v}\in C_{\rm div}(\Omega)\}\\
        H^m_{0,per}(\Omega)&=\{\th\in H^m(\Omega)\hbox{ vanishing at $0,d$ and periodic on $\partial D$ for every $z\in ]0,d[$}\}.
	\end{aligned}
	$$
   For an introduction to spaces of periodic $H^m$ functions see \cite{temam1997infinite}, p.~50. Recall the relation
	$$L^2(\Omega)=J^2(\Omega)\oplus G^2(\Omega)$$
	which is, for regular vector fields, the usual decomposition of ${\bf v}$ as a sum of a divergence-free vector field and a vector orthogonal to the space of solenoidal fields. In view of de Rham's theorem \cite{temam2001navier}, we have 
	$$
	\begin{aligned}
		G^2(\Omega)&=\{{\bf v}\in L^2(\Omega):{\bf v}=\nabla p\hbox{ for some }p\in H^1_{\rm loc}(\Omega)\}\\
		G^{-m,2}(\Omega)&=\{{\bf v}\in H^{-m}(\Omega):{\bf v}=\nabla p\hbox{ for some }p\in H^{-m+1}_{\rm loc}(\Omega)\}.
	\end{aligned}
	$$
	
	Let now ${\bf v}\in H^2_0(\Omega)$. In \cite{brezis2011functional} it is proved that there exist $C>0$ and $\alpha,\beta\in]0,1[$ such that
	\begin{equation}
		\label{eq:stimeBrez}
		\begin{aligned}
			||{\bf v}||_6&\leq C(||{\bf v}||_2+||{\bf v}||_2^\alpha||\Delta {\bf v}||_2^{1-\alpha})\\
			||\nabla {\bf v}||_3&\leq C(||{\bf v}||_2+||{\bf v}||_2^{\beta}||\Delta {\bf v}||_2^{1-\beta}).
		\end{aligned}
	\end{equation}
	Moreover, being ${\bf v}$ and $ \nabla {\bf v}$ zero at the boundary, it follows
	$$\int_\Omega|\nabla\nabla {\bf v}|^2\,dx=\int_\Omega v_{,ij}v_{,ij}\,dx=\int_\Omega v_{,ii}v_{,jj}\,dx=\int_\Omega|\Delta {\bf v}|^2\,dx$$
	and clearly
	$$\int_{\Omega}|\nabla {\bf v}|^2\,dx=-\int_\Omega {\bf v}\cdot \Delta {\bf v}\,dx\leq ||{\bf v}||_2||\Delta {\bf v}||_2.$$
	This implies that $(||{\bf v}||_2^2+||\Delta {\bf v}||_2^2)^{1/2}$ is an equivalent norm on $H^2_0(\Omega)$. If $\Omega$ is bounded, by a repeated application of Poincaré inequality, the same holds for $||\Delta {\bf v}||^2$ only. The same facts hold for scalar or vector fields with components in $H^2_0(\Omega)$ or $H^2_{0,per}$.
	
	\medskip
	
	We set ${\bf U}=({\bf v},\th)$ and rewrite the equations in compact form
	
	\begin{equation}
		\de{{\bf U}}{t}+N({\bf U})=L{\bf U}+{\bf G}
		\label{eq:eqastr}
	\end{equation}
	where
	$$N({\bf U})=
	\begin{bmatrix}
		{\bf v}\cdot \nabla &0\\
		0&{\bf v}\cdot\nabla 
	\end{bmatrix}
	\begin{bmatrix}
		{\bf v}\\
		\th
	\end{bmatrix},
	$$
	$$L{\bf U}=\begin{bmatrix}
		\nu\Delta-\xi^2\Delta^2&{\bf L}_1(z)\\
		{\bf L}_2(z)\cdot&\kappa_1\Delta-\kappa_2\Delta^2
	\end{bmatrix}\begin{bmatrix}
		{\bf v}\\
		\th
	\end{bmatrix}
,
	\qquad
	{\bf G}=\begin{bmatrix}
		-\frac{1}{\rho_0}\nabla p +{\bf G}_1(z)\\
		G_2(z)
	\end{bmatrix}
	$$
	and finally with $\div {\bf v}=0$, where ${\bf U}$ lies in the space $X=J^{2,2}(\Omega)\times H^2_{0,per}$, with the norm 
	$$||{\bf U}||_X^2=||\nabla\nabla {\bf v}||_2^2+||\nabla\nabla\th||_2^2=||\Delta {\bf v}||_2^2+||\Delta\th||_2^2$$
	which is in this case equivalent to the natural one. 
	
	\subsection{Estimates on the nonlinear terms}
	
	We first need some inequalities involving the nonlinear term $N({\bf U})$ in \eqref{eq:eqastr}.
	
	\begin{prop}
		Let ${\bf U}_1,{\bf U}_2\in H^2_0(\Omega)\times H^2_{0,per}$ with $\div {\bf v}_1=\div {\bf v}_2$ in $\Omega$ and $\th_1=\th_2$ on $\partial\Omega$. Then there exist $\gamma>0$ and, for any $\eps>0$, a number $C_\eps>0$ independent of ${\bf U}_1,{\bf U}_2$ such that
		\begin{equation}
			\langle N({\bf U}_1)-N({\bf U}_2),{\bf U}_1-{\bf U}_2\rangle\geq -\eps||\Delta ({\bf U}_1-{\bf U}_2)||_2^2-C_\eps(1+||{\bf U}_1||_2+||{\bf U}_2||_2)^\gamma||{\bf U}_1-{\bf U}_2||_2^2
			\label{eq:stimeNU}
		\end{equation}
		where $||{\bf U}||_2^2$ stands for $||{\bf v}||_2^2+||\th||_2^2$.
	\end{prop}
	
	\begin{proof} For the nonlinearity in the velocity, using H\"older inequality (see also \cite{Degiovanni:2020}) the following is not difficult to prove:
	\begin{equation}
		\begin{aligned}
			\int_\Omega((D{\bf v}_1){\bf v}_1-(D{\bf v}_2){\bf v}_2)&({\bf v}_1-{\bf v}_2)\,dx\geq\\
			&-C_1(||{\bf v}_1||_2+||{\bf v}_2||_2)||D({\bf v}_1-{\bf v}_2)||_3||{\bf v}_1-{\bf v}_2||_6
		\end{aligned}
		\label{eq:nonlfluid}
	\end{equation}
	for a suitable constant $C_1$ depending only on the dimension of space and	for every ${\bf v}_1,{\bf v}_2\in H^2_0(\Omega)$ with $\div {\bf v}_1=\div {\bf v}_2$.
	
	For the nonlinearity in the temperature we have instead
	$$
	\begin{aligned}
		&\int_\Omega({\bf v}_1\cdot\nabla\th_1-{\bf v}_2\cdot\nabla\th_2)(\th_1-\th_2)\,dx=\int_\Omega(\nabla(\th_1-\th_2)\cdot {\bf v}_1)(\th_1-\th_2)\,dx\\
		&+\int_\Omega\nabla\th_2\cdot({\bf v}_1-{\bf v}_2)(\th_1-\th_2)\,dx.
	\end{aligned}
	$$
	Now, if $\div {\bf v}_1=\div {\bf v}_2$ and since $\th_1,\th_2$ coincide on the boundary, the last integral is equal to
	$$\int_\Omega\div(\th_2({\bf v}_1-{\bf v}_2))(\th_1-\th_2)\,dx=-\int_\Omega\th_2({\bf v}_1-{\bf v}_2)\cdot\nabla(\th_1-\th_2).$$
	By H\"older inequality with exponents 2,3,6 we then get
	$$\begin{aligned}
		&\int_\Omega({\bf v}_1\cdot\nabla\th_1-{\bf v}_2\cdot\nabla\th_2)(\th_1-\th_2)\,dx\geq-C_2||{\bf v}_1||_2||\nabla(\th_1-\th_2)||_3||\th_1-\th_2||_6\\
		&-C_3||\th_2||_2||\nabla(\th_1-\th_2)||_3||{\bf v}_1-{\bf v}_2||_6
	\end{aligned}
	$$
	for suitable constants $C_2,C_3$. Switching the role of ${\bf v}_1,{\bf v}_2$ and $\th_1,\th_2$ and summing up, it is not difficult to see that there exists $C_4>0$ such that
	\begin{equation}
		\begin{aligned}
			&\int_\Omega({\bf v}_1\cdot\nabla\th_1-{\bf v}_2\cdot\nabla\th_2)(\th_1-\th_2)\,dx\geq\\
			&-C_4||\nabla(\th_1-\th_2)||_3((||{\bf v}_1||_2+||{\bf v}_2||_2)||\th_1-\th_2||_6+(||\th_1||_2+||\th_2||_2)||{\bf v}_1-{\bf v}_2||_6).
		\end{aligned}
		\label{eq:nonltemp}
	\end{equation}
	From \eqref{eq:nonlfluid} and \eqref{eq:nonltemp} it is easy to see that there exists $C_5>0$ such that
	$$
	\langle N({\bf U}_1)-N({\bf U}_2),{\bf U}_1-{\bf U}_2\rangle\geq -C_5(||{\bf U}_1||_2+||{\bf U}_2||_2)||D({\bf U}_1-{\bf U}_2)||_3||{\bf U}_1-{\bf U}_2||_6.
	$$
	By this and a repeated application of the estimates \eqref{eq:stimeBrez} the thesis follows.
	\end{proof}
	
	We remark that ${\bf U}_1,{\bf U}_2$ need not yet to be solutions of our problem.
	
	\subsection{The stationary case}
	
	We now want to recast our problem into the theory of variational inequalities to apply general existence theorems.
	
	We begin by setting, with $R>0$,
	$$K_R=\{{\bf z}\in X:||{\bf z}||_2\leq R\},\qquad \widehat K_R=\{{\bf z}\in H^2_0(\Omega)\times H^2_{0,per}(\Omega):||{\bf z}||_2\leq R\}.$$
	
	We then introduce $F:K_R\to H^{-2}(\Omega)\times H^{-2}(\Omega)$ defined as
	$$F({\bf U})=
	\begin{bmatrix}
		-\nu\Delta {\bf v}+\hat{\xi}\Delta^2 {\bf v}-{\bf L}_1(z)\th-{\bf G}_1(z)+\frac{1}{\rho_0} \nabla p+(\nabla {\bf v}){\bf v}\\
		-\kappa_1\Delta \th+\kappa_2\Delta^2\th-{\bf L}_2(z)\cdot {\bf v}-F_2(z)+{\bf v}\nabla\th
	\end{bmatrix}=N({\bf U})-L{\bf U}-{\bf G}.
	$$ 
	
	\begin{prop}\label{prop:maxmon}
		There exist $\delta_1,\delta_2>0$ independent of $R$ and $\omega_R>0$ depending on $R$ such that
		$$\langle F({\bf U}_1)-F({\bf U}_2),{\bf U}_1-{\bf U}_2\rangle\geq \delta_1||\Delta({\bf U}_1-{\bf U}_2)||_2^2+\delta_2||\nabla ({\bf U}_1-{\bf U}_2)||_2^2-\omega_R||{\bf U}_1-{\bf U}_2||_2^2,$$
  for all $ {\bf U}_1,{\bf U}_2\in K_R$ with $\operatorname{div}{\bf v}_1=\operatorname{div}{\bf v}_2$.
	\end{prop}
	
	\begin{proof}
	The pressure in the term ${\bf G}$ gives a zero contribution since ${\bf v}_1,{\bf v}_2\in J^2(\Omega)$ while the constant terms ${\bf G}_1,G_2$ drop out and then clearly
	$$\langle {\bf L}_1(z)(\th_1-\th_2),{\bf v}_1-{\bf v}_2\rangle\geq -C_6(||\th_1-\th_2||_2^2+||{\bf v}_1-{\bf v}_2||_2^2)=-C_7||{\bf U}_1-{\bf U}_2||_2^2$$
	for a constant $C_7>0$ depending only on the thickness $d$. A similar result holds for ${\bf L}_2$. The remaining terms in $\langle L{\bf U},{\bf U}\rangle$ give the positively dissipative terms
	$$\nu||\nabla ({\bf v}_1-{\bf v}_2)||_2^2+\hat{\xi}||\Delta ({\bf v}_1-{\bf v}_2)||_2^2+\kappa_1||\nabla (\th_1-\th_2)||_2^2+\kappa_2||\Delta (\th_1-\th_2)||_2^2$$
	and finally the remainder satisfies \eqref{eq:stimeNU}. Taking $\eps$ small enough in \eqref{eq:stimeNU} and remembering that $||{\bf U}_1||,||{\bf U}_2||\leq R$, the result easily follows. 
 \end{proof}
	
	\medskip
	
	At this point we can associate to our stationary problem a general variational inequality with a given right-hand side $f$. Namely (see \cite{Degiovanni:2020}, Theorem 7.4.1), given $R>0$, for every $f\in H^{-2}(\Omega)$ there exists only one ${\bf U}\in K_R$ such that
	$$
	\langle F({\bf U}),{\bf V}-{\bf U}\rangle+\omega_R\int_\Omega {\bf U}\cdot({\bf V}-{\bf U})\,dx+\int_\Omega {\bf U}\cdot({\bf V}-{\bf U})\,dx\geq\langle f,{\bf V}-{\bf U}\rangle
	$$
	for every ${\bf V}\in \widehat K_R$ with $\div {\bf v}=\div {\bf u}$. Introducing the \emph{normal cone} $N_{K_R}({\bf U})$ to $K_R$ at ${\bf U}$
	$$N_{K_R}({\bf U})=\{f\in H^{-2}(\Omega)\times H^{-2}(\Omega):\langle f,{\bf V}-{\bf U}\rangle\leq0\hbox{ for all }{\bf V}\in K_R\},$$
	this last result is equivalent to the existence of a unique solution of the differential inclusion
	$$F({\bf U})+\omega_R {\bf U}+{\bf U}+N_{K_R}({\bf U})\ni f.$$
	In \cite{Degiovanni:2020} the normal cone turned out to contain the pressure gradient term (since the point ${\bf v}$ of ``least distance'' from $f$ to $J^2(\Omega)$ is perpendicular to $J^2(\Omega)$ and therefore the cone is made up by a term in $G^{-2,2}(\Omega)$). Here the situation is similar but with one more variable $\th$; however, since $\th\in H^2_0(\Omega)$, the second component of the normal cone is zero. More precisely, adapting Proposition 7.4.3 of \cite{Degiovanni:2020} to our case, it is easily seen that whenever $||{\bf U}||<R$, the normal cone $N_{K_R}({\bf U})$ is given by
	$$N_{K_R}({\bf U})=G^{-2,2}(\Omega)\times\{0\}.$$
	
	\subsection{Existence and uniqueness results}
	
	At this point only energy estimates are needed to prove that a solution of the nonstationary problem exists and is unique in $J^2(\Omega)\times H^2_0(\Omega)$. The technique, however, relies on some definitions of maximal monotone operators that we recall briefly.\\
 From Proposition \ref{prop:maxmon} it follows now that $A+\omega_R I$ is a maximal monotone operator in $X$ (see \cite{Degiovanni:2020}, Theorem 7.4.3) and the whole existence theory given therein applies to our case provided ${\bf U}\in K_R$, i.e. $||{\bf U}||_2\leq R$. 
	
	Set
	$$D(A)=\{{\bf U}\in K_R:(\Delta^2 {\bf v},\Delta^2\th)\in [J^2(\Omega)\oplus G^{-2,2}(\Omega)]\times L^2_{0,per}(\Omega)\}$$
	and define a multivalued operator
	$$A({\bf U})=[F({\bf U})+N_{K_R}({\bf U})]\cap X,$$
	remembering that whenever a strong solution exists, $N_{K_R}({\bf U})$ is actually a singleton and hence $A$ is a usual differential operator.
	
	\begin{defn} Let $\Omega=$D$\times[0,d]$ where $D$ is a bounded set in $\R^2$ with regular boundary. We say that a function ${\bf U}$ is a \emph{strong solution} if \eqref{eq:eqastr} holds and ${\bf U}:[0,+\infty[\to X$ is continuous, if its restriction to $]0,+\infty[$ is absolutely continuous on compact sets, if ${\bf U}(t)\in J^{2,2}\times H^2$ for a.e. $t>0$ and finally if
		$${\bf U}'+L{\bf U}+N({\bf U})+{\bf G}\in G^{-2,2}(\Omega)\times\{0\}$$
		for a.e. $t>0$. This implies that there exists $p\in H^{-2}_{loc}(\Omega)$ such that \eqref{eq:semiastr} holds in $H^{-2}(\Omega)$. 
	\end{defn}
	
	Now we fix $T>0$, we multiply \eqref{eq:eqastr} by ${\bf U}$ in $X$, integrate on $[0,T]$ and notice that, due to the fact that $\langle N({\bf U}),{\bf U}\rangle=0$, we have
	$$\frac{\hfill d}{dt}||{\bf U}||_2^2=2\langle {\bf U}',{\bf U}\rangle\leq -2\delta_1||\nabla {\bf U}||_2^2-2\delta_2||\Delta {\bf U}||_2^2+||G||_2||{\bf U}||_2$$
	so that, using H\"older and Poincaré inequalities and integrating between $0$ and $T$ it follows
	\begin{equation}
		||{\bf U}(T)||^2+\delta_1\int_0^T||\Delta {\bf U}||_2^2(s)\,ds\leq ||{\bf U}(0)||_2^2+MC_dT
		\label{eq:energ1}
	\end{equation}
	where $M$ is a positive constant depending on the horizontal domain $D$ and $C_d$ depends only on the thickness $d$ of the slab. From this it follows that, whenever $U$ exists, it will satisfy $||{\bf U}(T)||_2^2\leq||{\bf U}_2(0)||_2^2+MC_dT$.
	
	Uniqueness follows now easily. Let $T>0$ and ${\bf U}_1,{\bf U}_2$ two strong solutions and let $R$ be such that 
	$$R^2\geq\max\{||{\bf U}_1(0)||_2^2+MC_dT,||{\bf U}_2(0)||_2^2+MC_dT\}.$$ 
	If ${\bf U}_1,{\bf U}_2$ are two strong solutions, then from \eqref{eq:energ1}
	$$||{\bf U}_i||_2^2(T)\leq ||{\bf U}_i(0)||_2^2+MC_dT\leq R^2\qquad (i=1,2)$$
	so that from Proposition \ref{prop:maxmon}
	$$
	\begin{aligned}
		\frac{\hfill d}{dt}||{\bf U}_1-{\bf U}_2||_2^2&=2\langle {\bf U}_1'-{\bf U}'_2,{\bf U}_1-{\bf U}_2\rangle =-2\langle F({\bf U}_1)-F({\bf U}_2),{\bf U}_1-{\bf U}_2\rangle\leq\\
		&\leq 2\omega_R||{\bf U}_1-{\bf U}_2||^2
	\end{aligned}
	$$
	which by integration implies ${\bf U}_1(T)={\bf U}_2(T)$ if ${\bf U}_1(0)={\bf U}_2(0)$. 
	
	The proof of existence now follows the one in \cite{Degiovanni:2020} and gives the following result.
	\begin{thm}
		For every ${\bf U}_0\in X$ there exists one and only one strong solution ${\bf U}(t)=({\bf u}(t),\th(t))$ of \eqref{eq:eqastr} such that ${\bf U}(0)={\bf U}_0$. Moreover, ${\bf U}(t)\in J^{2,2}(\Omega)\times H^2_0(\Omega)$, $\Delta^2 {\bf u}(t)\in J^2(\Omega)\oplus G^{-2,2}(\Omega)$ and $\Delta^2\th(t)\in L^2_{0,per}(\Omega)$ for all $t>0$, and for every $t_0>0$ the function ${\bf V}(t)={\bf U}(t+t_0)$ is the strong solution of \eqref{eq:eqastr} with ${\bf V}(0)={\bf U}(t_0)$.
	\end{thm} 

    \begin{rem}
       If the initial data are more regular (say, $H^4(\Omega)$) and the boundary is more regular too, then it can be proved that ${\bf U}$ also belongs to $H^4(\Omega)$ for all $t\in[0,T]$ and so is a classical solution of system \eqref{eq:sist}.
    \end{rem}
 
	\begin{rem} {\rm The request that $D$ must be bounded is essentially due to the fact that the solution $\bar T$ introduced above is unbounded in $L^2(D)$. If $\bar T$ is replaced by any solution bounded in $L^2(\Omega)$, then existence and uniqueness follow also for a general domain $D$.}
	\end{rem}

	\section{Thermal convection} 
	\label{S:Convection}
	
	We now suppose equations \eqref{E:Mom}, \eqref{E:Cty} and \eqref{E:Teq} hold in the horizontal layer
	$\{(x,y)\in\mathbb{R}^2\}\times\{0<z<d\}$
	for
	$t\geq 0$
	with gravity acting in the negative
	$z$ direction.
	Thus,
	$g_i=-k_i{\bf g}$,
	where
	${\bf k}=(0,0,1)$.
	The temperatures of the upper and lower planes are kept fixed at 
	$T=T_L$
	at
	$z=0$, $T=T_U$
	at
	$z=d$,
	where
	$T_L,T_U$
	are constants with
	$T_L>T_U$.
	In this case 
	the system of equations \eqref{E:Mom}, \eqref{E:Cty} and \eqref{E:Teq} possesses the steady conduction unique solution
	\begin{equation*}
		{\bar v}_i\equiv 0,\qquad{\bar T}={\bar T}(z),\qquad {\bar p}={\bar p}(z),
	\end{equation*}
	where
	${\bar T}$
	solves
	\begin{equation*}
		\kappa_1\frac{d^2{\bar T}}{dz^2}-\kappa_2\frac{d^4{\bar T}}{dz^4}=0.
	\end{equation*}
	The steady pressure 
	${\bar p}(z)$
	is then found from the steady momentum equation, up to a constant. 
 
 The boundary conditions of strong adherence advocated by
	\cite{FriedGurtin:2006} correspond to 
	$v_i=0,\partial v_i/\partial {\bf n}=0$
	on the horizontal boundaries 
	$z=0,d$,
	and we suppose the solution is periodic in
	$x,y$.
	For thermal convection we suppose the solution as a function of
	$x$
	and
	$y$
	satisfies a horizontal planform which tiles the plane. In particular, a hexagonal planform which is observed in real life,
	is discussed in detail in
	\cite[pages~43-52]{Chandrasekhar:1981}.
 
	The temperature on
	$z=0,d$
	is known and we also suppose
	$\partial T/\partial z=0$
	there. This then yields the steady temperature field in \eqref{tbar}, that it's convenient to rewrite as
	\begin{equation}\label{E:TbarSS}
		{\bar T}(z)=T_L+c_2(\sinh \lambda z-c_1\cosh \lambda z-\lambda z+c_1)\,,
	\end{equation}
	where
	$\beta=(T_L-T_U)/d\,>0$,
        $\lambda=\sqrt{\kappa_1/\kappa_2},$
        and
\begin{equation*}
c_1=\frac{\cosh\lambda d-1}{\sinh\lambda d}\,,\qquad 
c_2=\frac{\beta}{\lambda-2c_1/d}\,.
\end{equation*}
	
	To analyse the stability of the steady solution
	we introduce perturbations 
	$(u_i,\theta,\pi)$ 
	to
	$({\bar v}_i,{\bar T},{\bar p})$
	by
	\begin{equation*}
		v_i={\bar v}_i+u_i,\quad 
		T={\bar T}+\theta,\quad 
		p={\bar p}+\pi\,.
	\end{equation*}
	
	The perturbation equations for
	$(u_i,\theta,\pi)$ 
	are derived and we non-dimensionalize with the scales
	\begin{align*}
		&u_i=u_i^*U,\qquad x_i=x_i^*d,\qquad t=t^*{\mathcal T},\qquad{\mathcal T}=\frac{d}{U},\qquad\xi=\frac{{\hat\xi}}{\nu d^2}\,,\\ 
		&\theta=\theta^*T^{\sharp},\qquad \pi=\pi^*P, \qquad P=\frac{\rho_0\nu U}{d}\,,\qquad 
		\kappa=\frac{{\kappa_2}}{d^2\kappa_1}\,,\\
		&T^{\sharp}=U\sqrt{\frac{\beta\nu}{\kappa_1\alpha g}}\,.
	\end{align*}
	The Rayleigh number 
	$Ra$
	is defined as
	\begin{equation*}
		Ra=R^2=\frac{\alpha\beta gd^4}{\kappa_1\nu}\,.
	\end{equation*}
	
	The non-dimensional perturbation equations are (where we omit *s), with $\nu = d U$,
	\begin{equation}\label{E:Pertn}
		\begin{aligned}
			u_{i,t}+u_ju_{i,j}&=-\pi_{,i}+R\theta k_i+\Delta u_i-\xi\Delta^2u_i,\\
			u_{i,i}&=0,\\
			Pr(\theta_{,t}+u_i\theta_{,i})&=f(z)Rw+\Delta\theta-\kappa\Delta^2\theta
		\end{aligned}
	\end{equation}
	where
	$Pr=\nu/\kappa_1$
	is the Prandtl number, and
 \begin{equation*}
		f(z)=c_4(1-\cosh Az+c_3\sinh Az)\,,
	\end{equation*}
where
$A=1/\sqrt{\kappa}$,
and
\begin{equation*}
c_3=\frac{\cosh A-1}{\sinh A}\,,\qquad c_4=\frac{A}{A-2c_3}\,.
\end{equation*}
	Equations \eqref{E:Pertn} hold on the domain
	$\{(x,y)\in \mathbb{R}^2\}\times\{z\in(0,1)\}$
	for 
	$t>0$.
	The boundary conditions are
	\begin{equation*}
		u_i=0,\quad \frac{\partial u_i}{\partial z}=0,\quad \theta=0,\quad \frac{\partial\theta}{\partial z}=0,\qquad z=0,1,
	\end{equation*}
	together with periodicity in
	$x,y$.

	\section{Linear instability theory} 
	\label{S:Instability}

	To find the critical Rayleigh numbers for instability we linearize \eqref{E:Pertn} and look for solutions of the form
	\begin{equation*}
		u_i=u_i({\bf x})e^{\sigma t},\qquad
		\theta=\theta({\bf x})e^{\sigma t},\qquad
		\pi=\pi({\bf x})e^{\sigma t}.
	\end{equation*}
	We then remove the pressure
	$\pi$
	by taking curlcurl of \eqref{E:Pertn}$_1$ and we retain the third component of the result. This leads to solving the eigenvalue
	problem for
	$\sigma$,
	namely
	\begin{equation}\label{E:Lin}
		\begin{aligned}
			&\sigma \Delta w=\Delta^2w-\xi\Delta^3w+R\Delta^*\theta,\\
			&\sigma Pr\theta=f(z)Rw+\Delta\theta-\kappa\Delta^2\theta,
		\end{aligned}
	\end{equation}
	where 
	$\Delta^*=\partial^2/\partial x^2+\partial^2/\partial y^2$.
	To solve this  we write 
	\begin{equation*}
		w=\frac{W(z)h(x,y)}{R},\quad\theta=\Theta(z)h(x,y),
	\end{equation*}
	where
	$h$
	is the planform discussed in
	\cite[pages~43-52]{Chandrasekhar:1981}, which satisfies
	$\Delta^*h=-a^2h$,
	where
	$a$
	is a wavenumber.
	Let
	$D=d/dz$,
	and then we rewrite \eqref{E:Lin} to reduce the analysis to solving the system
	\begin{equation}\label{E:LinD2}
		\begin{aligned}
			&(D^2-a^2)W-\chi=0,\\
			&(D^2-a^2)\chi-\psi=0,\\
			&(D^2-a^2)\psi-\frac{\psi}{\xi}+\frac{Ra}{\xi}a^2\Theta=-\frac{\sigma}{\xi}\chi,\\
			&(D^2-a^2)\Theta-\Phi=0,\\
			&(D^2-a^2)\Phi-\frac{\Phi}{\kappa}-\frac{f(z)}{\kappa}W=-\frac{\sigma}{\kappa}Pr\Theta,
		\end{aligned}
	\end{equation}
	for
	$z\in(0,1)$,
	together with
	the boundary conditions
	\begin{equation*}
		W=DW=D^2W=\Theta=D\Theta=0,\quad{\rm on}\,\,z=0,1.
	\end{equation*}
	
	Note that we have transformed \eqref{E:Lin} to rearrange 
	$R$
	as
	$Ra$
	in \eqref{E:Lin}$_1$.
	The boundary conditions are also conveniently rewritten as
	\begin{equation*}
		W=0,\quad \chi=0,\quad DW=0,\quad \Theta=0,\quad D\Theta=0,\qquad{\rm on}\,\, z=0,1.
	\end{equation*}
	
	To solve this system numerically the solution is written as a sum of Chebyshev polynomials of form
	\begin{align*}
		&W=\sum_{i=0}^NW_iT_i(z),\\
		&\chi=\sum_{i=0}^N\chi_iT_i(z),\\
		&\psi=\sum_{i=0}^N\psi_iT_i(z),\\
		&\Theta=\sum_{i=0}^N\Theta_iT_i(z),\\
		&\Phi=\sum_{i=0}^N\Phi_iT_i(z).
	\end{align*}
	The discrete version of system \eqref{E:LinD2} gives rise to a generalized matrix eigenvalue problem of form
	\begin{equation}\label{E:AB}
		A{\bf x}=Ra\,B{\bf x}
	\end{equation}
	where
	\begin{equation}\label{E:LinBCs}
		{\bf x}=(W_0,\ldots,W_N,\chi_0,\ldots,\chi_N,\psi_0,\ldots,\psi_N,\Theta_0,\ldots,\Theta_N,\Phi_0,\cdots,\Phi_N)
	\end{equation}
	and the boundary conditions are incorporated into the matrix
	$A$
	by writing them into the appropriate rows of
	$A$.
	The generalized matrix eigenvalue problem \eqref{E:AB} is then solved for the eigenvalues 
	$\sigma$
	by the QZ algorithm of \cite{MolerStewart:1971}.\\
To avoid round off problems with subtracting large but nearly equal numbers we rewrite 
$f(z)$
in the numerical code as
\begin{equation*}
f(z)=\displaystyle{\frac{1-e^{-Az}-(1-e^{-A})(\sinh Az/\sinh A)}{1-(2\sinh(A/2)/(A\cosh(A/2)))}}\,.
\end{equation*}
The function
$f(z)$
is expanded as a series in
$T_n(z)$
and the Fourier coefficients are then used to calculate the matrix
$f(z)*w$,
cf. \cite[~page 829]{PayneStraughan:2000}. We actually solve the numerical system with 
$\sigma\in\mathbb{R}$
since the fully nonlinear stability values found by energy stability theory are so close to the linear instability ones it is practically impossible
for oscillatory convection and/or sub-critical instabilities to be important.

While we have not been able to show analytically that the eigenvalues of \eqref{E:Lin} are real under the boundary conditions \eqref{E:LinBCs}
it is of interest to note that one may do so for idealized boundary conditions. One has to be very careful with boundary conditions for \eqref{E:Lin}
as \cite{Ladyzhenskaya:2003}, \cite[Section 4]{Straughan:2023EPJP} demonstrate.

If we adopt illustrative boundary conditions as in 
\cite[Section 6.5]{Straughan:2023EPJP} then we may arrange \eqref{E:Lin} as
\begin{equation}\label{E:Sys1}
\begin{aligned}
&\mathcal{L}w\equiv \sigma\Delta w-\Delta^2w+\xi\Delta^3w=-R a^2\theta,\\
&\mathcal{M}\theta\equiv \sigma Pr\theta -\Delta \theta+\kappa\Delta^2\theta =f(z)Rw,\\
\end{aligned}
\end{equation}
and suppose the boundary conditions are
\begin{equation}\label{E:BCs1}
w=0,\quad w_{zz}=0,\quad w_{zzzz}=0,\quad\theta=0,\quad\theta_{zz}=0,
\end{equation}
on
$z=0,1$.
The linear operators
$\mathcal{L}$
and
$\mathcal{M}$
are as shown. One may eliminate 
$\theta$
and find the full equation for
$w$,
\begin{equation*}
\mathcal{M}\mathcal{L}w=-R^2a^2f(z)w.
\end{equation*}
Now multiply this equation by
$w^*$,
the complex conjugate of
$w$,
and integrate over a period cell
$V$.

This leads to the equation
\begin{align*}
-\sigma^2Pr\Vert\nabla w\Vert^2&-\sigma (Pr+1)\Vert\Delta w\Vert^2-\sigma (Pr\xi+\kappa)\Vert\nabla\Delta w\Vert^2\\
&-\Vert\nabla\Delta w\Vert^2-(\xi+\kappa)\Vert\Delta^2w\Vert^2
-\kappa\xi\Vert\nabla\Delta^2w\Vert^2=-R^2a^2(f(z)w,w^*).
\end{align*}
Put now
$\sigma=\sigma_r+i\sigma_1$
and take the imaginary part of this equation to find
\begin{equation*}
\sigma_12\sigma_rPr\Vert\nabla w\Vert^2=-\sigma_1\bigl[(Pr+1)\Vert\Delta w\Vert^2+(Pr\xi+\kappa)\Vert\nabla\Delta w\Vert^2\bigr].
\end{equation*}
If
$\sigma_1\ne 0$
then it follows 
$\sigma_r<0$
and the principle of exchange of stabilities holds.

	The resulting critical value of
	$Ra(a^2)$
	is then minimized in
	$a^2$ to find the linear instability value for each value of 
	$\xi,\kappa$.

	Numerical results are reported in Section \ref{S:Num}.

	\section{Global nonlinear stability} 
	\label{S:Stab}
	
	Linear instability theory yields a threshold for when the solution becomes unstable but this threshold does not guarantee that the
	solution will be stable if the Rayleigh number is below this value. 
 We now develop a nonlinear energy stability theory to yield
	a global (for all initial data) bound for nonlinear stability.
	
	To do this let
	$V$
	be a period cell for the solution to \eqref{E:Pertn} and let
	$\Vert\cdot\Vert$
	and
	$(\cdot,\cdot)$
	denote the norm and inner product on
	$L^2(V)$.
	Multiply \eqref{E:Pertn}$_1$ by
	$u_i$
	and integrate over
	$V$
	and likewise multiply \eqref{E:Pertn}$_3$ by
	$\theta$
	and integrate over
	$V$.
 
	After integration by parts and use of the boundary conditions one may find
	\begin{equation}\label{E:En1}
		\frac{d}{dt}\frac{1}{2}\Vert{\bf u}\Vert^2=R(\theta,w)-\Vert\nabla{\bf u}\Vert^2-\xi\Vert\Delta{\bf u}\Vert^2,
	\end{equation}
	and
	\begin{equation}\label{E:En2}
		\frac{d}{dt}\frac{Pr}{2}\Vert{\theta}\Vert^2=R(fw,\theta)-\Vert\nabla{\theta}\Vert^2-\kappa\Vert\Delta{\theta}\Vert^2.
	\end{equation}
	Let 
	$\lambda>0$
	be a coupling parameter to be chosen opportunely and form \eqref{E:En1}+$\lambda$\eqref{E:En2}.
	In this manner we obtain
	\begin{equation}\label{E:EnEq}
		\frac{dE}{dt}=RI-D,
	\end{equation}
	where
	\begin{equation*}
		E=\frac{1}{2}\Vert{\bf u}\Vert^2+\frac{\lambda Pr}{2}\Vert\theta\Vert^2\,,
	\end{equation*}
	the production term is
	\begin{equation*}
		I=(\theta,w(1+\lambda f))\,,
	\end{equation*}
	while the dissipation is
	\begin{equation*}
		D=\Vert\nabla{\bf u}\Vert^2+\xi\Vert\Delta{\bf u}\Vert^2
		+\lambda\Vert\nabla{\theta}\Vert^2+\lambda\kappa\Vert\Delta{\theta}\Vert^2.
	\end{equation*}
	
	Let us consider the space
	\[
	H=\{(\gr{u},\th)\in H^2(V)\times H^2(V):\ u_{i,i}=0\}
	\]
	restricted to periodicity conditions on $(x,y)$ and subjected to boundary conditions $u_i,\de{u_i}{{\bf n}},\th,\de{\th}{{\bf n}}=0$ on $z=0,1$ and consider the quantity
	\[
	\frac{I}{D}=\frac{(\th,[1+\lambda f]w)}{\|\nabla\gr{u}\|^2+
		\xi\|\Delta\gr{u}\|^2+\lambda\|\nabla\th\|^2+\lambda\kappa\|\Delta\th\|^2}.
	\]
	By the boundary conditions on $z=0,1$, the standard Poincaré
	inequality implies that $I/D$ is bounded, hence it makes sense to
	consider
	\begin{equation}
 \label{E:REDef}
	\frac{1}{R_E}=\sup_H\frac{I}{D}.
	\end{equation}
	We now prove that the supremum is indeed a maximum, following the method first introduced by Rionero in \cite{rionero1968metodi} and then used also in \cite{Galdi1985}. Since both $I$ and $D$ are quadratic, one has
	\[
	\sup_H\frac{I}{D}=\sup_{D=1}I.
	\]
	Taking a maximizing sequence $(\gr{u}^{(h)},\th^{(h)})$ with $D(\gr{u}^{(h)},\th^{(h)})=1$, that is
	\[
	I(\gr{u}^{(h)},\th^{(h)})\to\sup_{D=1}I,
	\]
	since the sequence is bounded in $H$ one has, up to a subsequence, that
	\[
	(\gr{u}^{(h)},\th^{(h)})\weak(\gr{u},\th)\quad\text{in $H$}
	\]
	and $D(\gr{u},\th)\leq 1$ by lower semicontinuity. Moreover, up to a subsequence,
	\[
	(\gr{u}^{(h)},\th^{(h)})\to(\gr{u},\th)\quad\text{in $L^2$},
	\]
	hence $I(\gr{u}^{(h)},\th^{(h)})\to I(\gr{u},\th)$. Then
	\[
	\frac{I(\gr{u},\th)}{D(\gr{u},\th)}\geq
	\lim_h \frac{I(\gr{u}^{(h)},\th^{(h)})}{D(\gr{u}^{(h)},\th^{(h)})} = 
	\lim_h I(\gr{u}^{(h)},\th^{(h)})=\sup_H\frac{I}{D}
	\]
	and $(\gr{u},\th)$ is a maximum point.
	
	From \eqref{E:EnEq}, similarly to what was done in \cite{Galdi1991MathematicalPF}, one sees that
	\begin{equation}\label{E:EnIneq}
		\frac{dE}{dt}\le -D\Bigl(1-\frac{R}{R_E}\Bigr)\,.
	\end{equation}
	Suppose that
	$R<R_E$
	so that
	$\gamma=1-R/R_E\,>\,0$,
	then from inequality \eqref{E:EnIneq} one may deduce
	\begin{equation}\label{E:EnIn2}
		\frac{dE}{dt}\le -k\gamma E,
	\end{equation}
	where
	\begin{equation*}
		k=\min\Bigl\{2\pi^2(1+\xi\pi^2),\frac{2\pi^2}{Pr}(1+\kappa\pi^2)\Bigr\}\,.
	\end{equation*}
	From \eqref{E:EnIn2}
	\begin{equation*}
		E(t)\le e^{-k\gamma t}E(0)
	\end{equation*}
	and so we obtain global nonlinear stability provided
	$R<R_E$.
	
	To find
	$R_E$
	we calculate the Euler-Lagrange equations from \eqref{E:REDef}, with the change of variables $\varphi=\sqrt{\lambda}\theta$. These are
	\begin{equation}\label{E:ELeqs}
		\begin{aligned}
			&R_EF\varphi k_i+\epsilon_{,i}+\Delta u_i-\xi\Delta^2u_i=0,\\
			&u_{i,i}=0,\\
			&R_EFw+\Delta\varphi-\kappa\Delta^2\varphi=0,
		\end{aligned}
	\end{equation}
	where
	$\epsilon$
	is a Lagrange multiplier and
	$F(z)=(1+\lambda f)/(2\sqrt{\lambda})$.
	To solve equations \eqref{E:ELeqs} we eliminate
	$\epsilon$
	to obtain
	\begin{equation}\label{E:NonEqs}
		\begin{aligned}
			&-R_EF\Delta^*\varphi-\Delta^2w+\xi\Delta^3w=0,\\
			&R_EFw+\Delta\varphi-\kappa\Delta^2\varphi=0.
		\end{aligned}
	\end{equation}
	System \eqref{E:NonEqs} is solved numerically by a Chebyshev tau-QZ algorithm method as in Section \ref{S:Instability}
	subjected to the boundary conditions
	\begin{equation*}
		w=w^{\prime}=w^{\prime\prime}=\varphi=\varphi^{\prime}=0,\qquad{\rm on}\,\, z=0,1.
	\end{equation*}
	We then determine the nonlinear stability thresholds
	\begin{equation*}
		Ra_E=\max_{\lambda>0}\min_{a^2>0}\,R_E^2(a^2,\lambda)\,.
	\end{equation*}
	Numerical results are reported in Section \ref{S:Num}.

	\section{Numerical results} 
	\label{S:Num}
	
	Numerical results are given in Tables \ref{Ta:Ta1} - \ref{Ta:Ta3} and Figures \ref{fig:fig1} - \ref{fig:fig4}.
	
	Figure \ref{fig:fig1} shows the behaviour of the critical Rayleigh number 
	$Ra$
	against
	$\kappa$
	for
	$\xi$
	fixed. Table \ref{Ta:Ta2} gives numerical values over a larger range of
	$\kappa$.
	In all cases 
	$Ra$
	increases with increasing
	$\kappa$
	(and
	$\xi$).
 
	Similar comments apply to the behaviour of
	$Ra$
	against
	$\xi$
	(for fixed
	$\kappa$)
	as shown in Figure \ref{fig:fig3} and Table \ref{Ta:Ta1}, although the actual
	$Ra$
	values are smaller for
	$\kappa$
	increasing. Thus, the stabilizing effect of the
	$\xi$
	term in the momentum equation is greater than the stabilizing effect of the
	$\kappa$
	term in the heat equation.
	
	The wavenumber behaviour as
	$\kappa$
	is increased (for fixed
	$\xi$)
	is shown in Figure \ref{fig:fig2}. This shows that increasing
	$\kappa$ has the effect of making the convection cells more narrow for 
	$\kappa$
	small
	but the wavenumber reaches a maximum and thereafter decreases. After reaching the maximum a further increase in
	$\kappa$
	leads to a relatively rapid widening of the cells. Thus, in this range
	$\kappa$
	increasing has the effect of increasing the critical Rayleigh number thereby making the layer more stable, but in some sense
	making the convection less intense as the cell width increases.
 
The actual maximum values of the wavenumber in Figure \ref{fig:fig2} are given by
\begin{align*}        
&{\rm For}\,\,\xi=10^{-4},\quad Ra=1979\quad a^2_{max}=10.333,\quad{\rm at}\,\,\kappa=4.3\times10^{-3},\\
&{\rm For}\,\,\xi=10^{-3},\quad Ra=2772\quad a^2_{max}=10.809,\quad{\rm at}\,\,\kappa=3.4\times10^{-3},\\
&{\rm For}\,\,\xi=10^{-2},\quad Ra=4119\quad a^2_{max}=11.728,\quad{\rm at}\,\,\kappa=2.45\times10^{-3}.
\end{align*}        
	
	Figure \ref{fig:fig4} and Table \ref{Ta:Ta1} show how the wavenumber increases with
	increasing
	$\xi$.
	Since the wavenumber is inversely proportional to the aspect ratio of the convection cell (width to depth ratio)
	this means that at the onset of thermal convection increasing
	$\xi$
	has the effect of narrowing the convection cells. Thus, the bi-Laplacian term in the momentum equation is in a sense intensifying
	the convection by making it occur in narrower cells.
	
	It is difficult to examine the behaviour of the solution as
	$\xi\to 0$
	or
	$\kappa\to 0$
	since the problem in each case becomes singular. In the case of classical B\'enard convection where
	$\xi=0$
	and
	$\kappa=0$
	there are no boundary conditions on
	$w^{\prime\prime}$
	and
	$\theta^{\prime}$
	and also the basic temperature profile is linear in
	$z$
	as opposed to being exponential.
	
	We have calculated the critical Rayleigh number from the fully nonlinear theory and in Table \ref{Ta:Ta3} we show a comparison of the results for the values of
	linear theory, denoted by
	$Ra$,
	and those for global nonlinear stability, indicated by
	$Ra_E$.
	It is seen that in all cases shown
	$Ra_E$
	is extremely close to
	$Ra$.
	In fact these values are so close that it is probably not possible to distinguish between them on an experimental scale.
	Thus, we may be reasonably confident that the results from linear instability theory are displaying a true picture of what one 
	will see. 
	 \begin{table}[htp]
		\begin{center}
			\begin{tabular}{|l|l|l|l|}
				\hline
				$Ra$        &     $a^2$    & $\xi$              & $\kappa$          \\ \hline
				2130.19     &     10.10    & $10^{-6}$          & $10^{-2}$         \\ \hline
				2254.58     &     10.27    & $10^{-4}$          & $10^{-2}$         \\ \hline
				2456.56     &     10.54    & $5\times 10^{-4}$  & $10^{-2}$          \\ \hline
				2635.05     &     10.72    & $10^{-3}$          & $10^{-2}$         \\ \hline
				3692.40     &     11.30    & $5\times 10^{-3}$  & $10^{-2}$          \\ \hline
				4856.59     &     11.57    & $10^{-2}$          & $10^{-2}$         \\ \hline\hline
				1739.12     &     10.03    & $10^{-6}$          & $10^{-3}$         \\ \hline
				1844.59     &     10.25    & $10^{-4}$          & $10^{-3}$         \\ \hline
				2015.41     &     10.55    & $5\times 10^{-4}$  & $10^{-3}$         \\ \hline
				2165.59     &     10.75    & $10^{-3}$          & $10^{-3}$         \\ \hline
				3048.24     &     11.39    & $5\times 10^{-3}$  & $10^{-3}$         \\ \hline
				4015.37     &     11.69    & $10^{-2}$          & $10^{-3}$         \\ \hline
			\end{tabular}
			\caption{
				Critical Rayleigh and wavenumbers for linear instability.
				Showing variation with $\xi$.
			}
			\label{Ta:Ta1}
		\end{center}
	\end{table}
	\begin{table}[htp]
		\begin{center}
			\begin{tabular}{|l|l|l|l|}
				\hline
				$Ra$        &     $a^2$    & $\kappa$           & $\xi$             \\ \hline
				3998.78     &     11.67    & $7\times 10^{-4}$  & $10^{-2}$         \\ \hline
				4015.37     &     11.69    & $10^{-3}$          & $10^{-2}$         \\ \hline
				4347.50     &     11.69    & $5\times 10^{-3}$  & $10^{-2}$         \\ \hline
				4856.59     &     11.57    & $10^{-2}$          & $10^{-2}$         \\ \hline
				9450.46     &     11.08    & $5\times 10^{-2}$  & $10^{-2}$         \\ \hline
				15359.0     &     10.91    & $0.1$              & $10^{-2}$         \\ \hline\hline
				2155.14     &     10.72    & $7\times 10^{-4}$  & $10^{-3}$         \\ \hline
				2165.59     &     10.75    & $10^{-3}$          & $10^{-3}$         \\ \hline
				2354.39     &     10.80    & $5\times 10^{-3}$  & $10^{-3}$         \\ \hline
				2635.05     &     10.72    & $10^{-2}$          & $10^{-3}$         \\ \hline\hline
				1834.72     &     10.22    & $7\times 10^{-4}$  & $10^{-4}$         \\ \hline
				1844.59     &     10.25    & $10^{-3}$          & $10^{-4}$         \\ \hline
				2011.61     &     10.33    & $5\times 10^{-3}$  & $10^{-4}$         \\ \hline
				2254.58     &     10.27    & $10^{-2}$          & $10^{-4}$         \\ \hline
				4404.14     &      9.97    & $5\times 10^{-2}$  & $10^{-4}$         \\ \hline
				7159.88     &      9.86    & $10^{-1}$          & $10^{-4}$         \\ \hline
			\end{tabular}
			\caption{
				Critical Rayleigh and wavenumbers for linear instability.
				Showing variation with $\kappa$.
			}
			\label{Ta:Ta2}
		\end{center}
	\end{table}
	\begin{table}[htp]
		\begin{center}
			\begin{tabular}{|l|l|l|l|l|l|l|}
				\hline
				$Ra$      &     $a^2$    & $Ra_E$  &  $a_E^2$  &  $\lambda$  &  $\kappa$        &     $\xi$             \\ \hline
				4856.59   &     11.57    & 4856.39 &  11.57    &    0.829    &  $10^{-2}$       &   $10^{-2}$           \\ \hline
				2165.59   &     10.75    & 2165.58 &  10.75    &    0.94     &  $10^{-3}$       &   $10^{-3}$           \\ \hline
				4015.37   &     11.69    & 4015.36 &  11.69    &    0.94     &  $10^{-3}$       &   $10^{-2}$           \\ \hline
				2635.05   &     10.72    & 2634.90 &  10.72    &    0.83     &  $10^{-2}$       &   $10^{-3}$           \\ \hline
				7159.88   &      9.86    & 7157.87 &   9.86    &    0.75     &  $10^{-1}$       &   $10^{-4}$           \\ \hline
			\end{tabular}
			\caption{
				Critical Rayleigh and wavenumbers for linear instability vs. those for nonlinear stability.
				Values of $\xi,\kappa$ shown, along with optimal value of $\lambda$.
			}
			\label{Ta:Ta3}
		\end{center}
	\end{table}
	\begin{figure}[htp]
		\begin{center}
\includegraphics[width=0.9\linewidth]{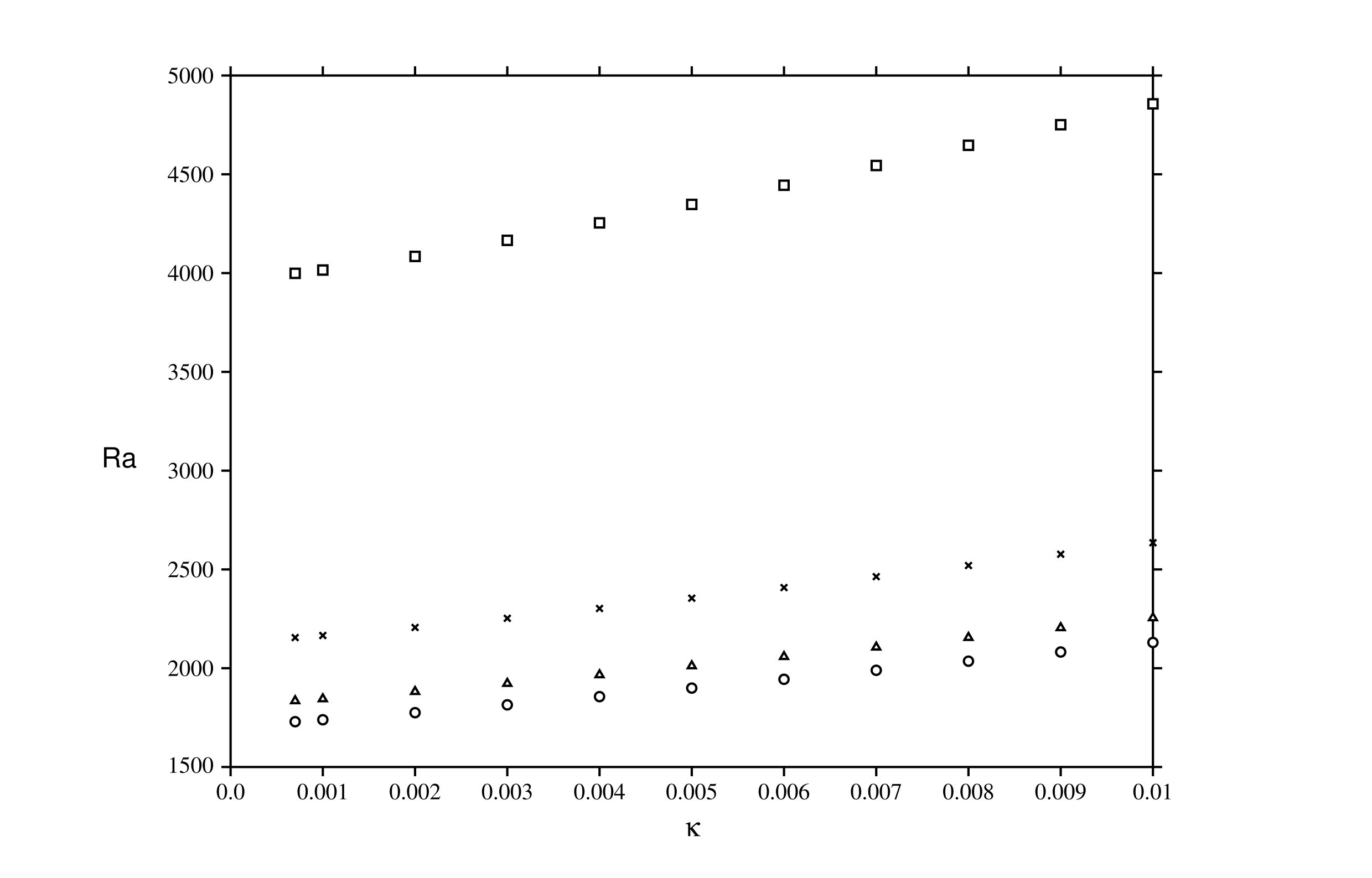}
			\caption{
				Graph of 
				$Ra$ 
				versus 
				$\kappa$.
				$\Box$ denotes $\xi=0.01$; 
				$\times$ denotes $\xi=10^{-3}$; 
				$\triangle$ denotes $\xi=10^{-4}$.
				$\circ$ denotes $\xi=10^{-6}$; 
				$\kappa$ runs from $0.0007$ to 0.01.
			}
			\label{fig:fig1}
		\end{center}
	\end{figure}
	\begin{figure}[htp]
		\begin{center}
\includegraphics[width=0.9\linewidth]{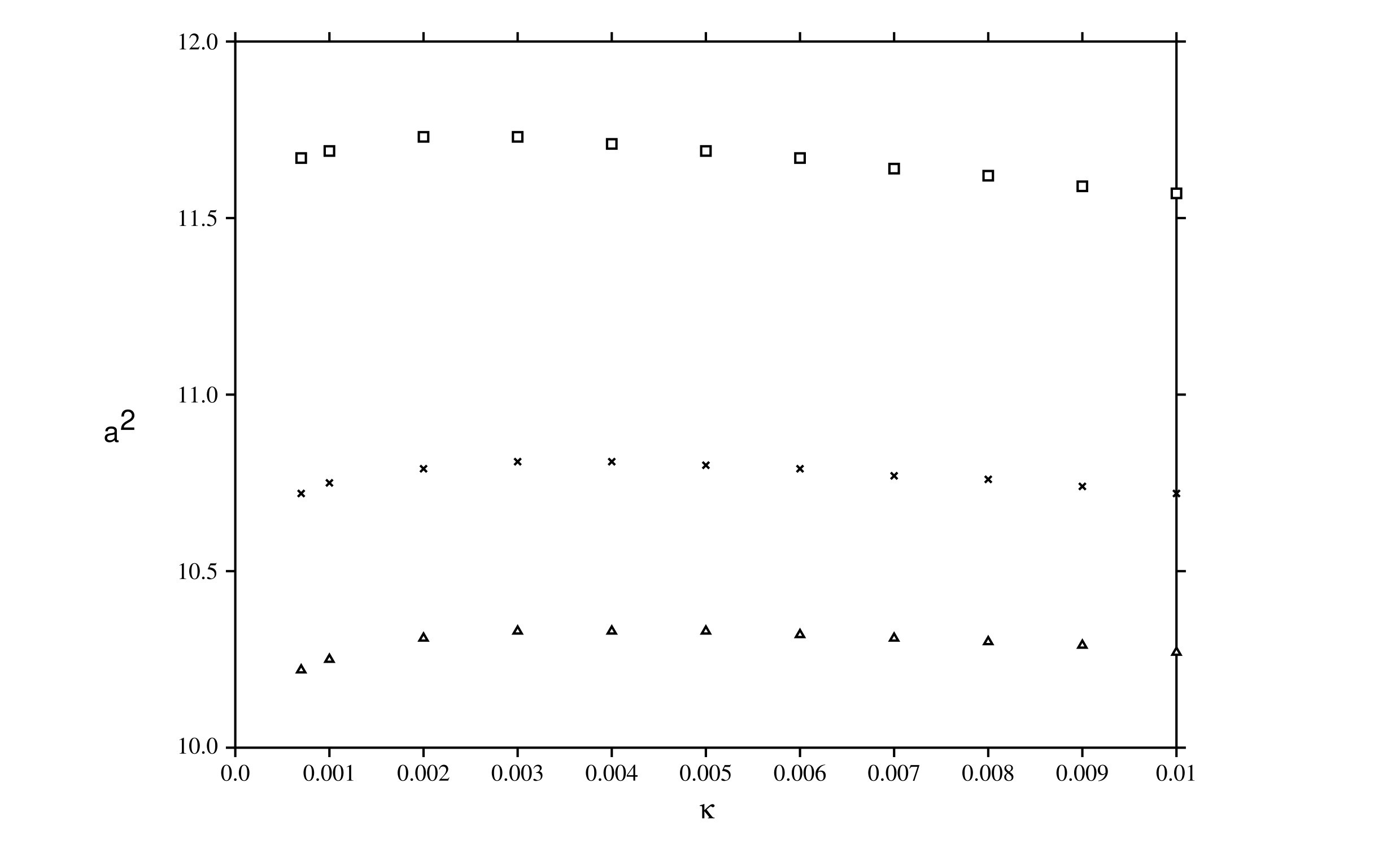}
			\caption{
				Graph of 
				$a^2$ 
				versus 
				$\kappa$. 
				$\Box$ denotes $\xi=0.01$; 
				$\times$ denotes $\xi=10^{-3}$; 
				$\triangle$ denotes $\xi=10^{-4}$.
				$\kappa$ runs from $0.0007$ to 0.01.
			}
			\label{fig:fig2}
		\end{center}
	\end{figure}

	\begin{figure}[htp]
		\begin{center}
			\centering\includegraphics[width=0.9\linewidth]{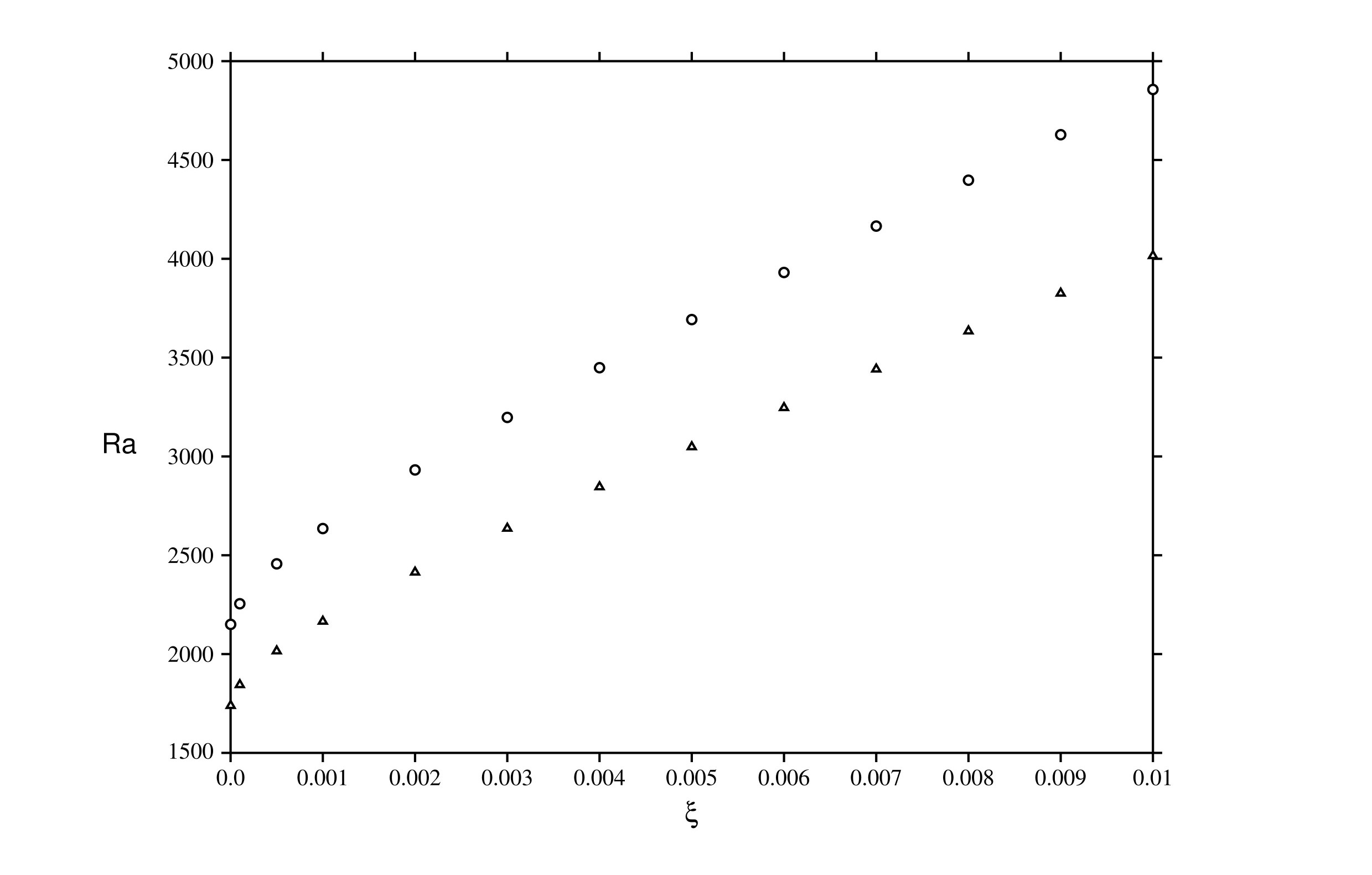}
			\caption{
				Graph of 
				$Ra$ 
				versus 
				$\xi$.
				$\circ$ denotes $\kappa=0.01$; 
				$\triangle$ denotes $\kappa=10^{-3}$; 
				$\xi$ runs from $10^{-6}$ to 0.01.
			}
			\label{fig:fig3}
		\end{center}
	\end{figure}
\begin{figure}[htp]
\begin{center}
\includegraphics[width=0.9\linewidth]{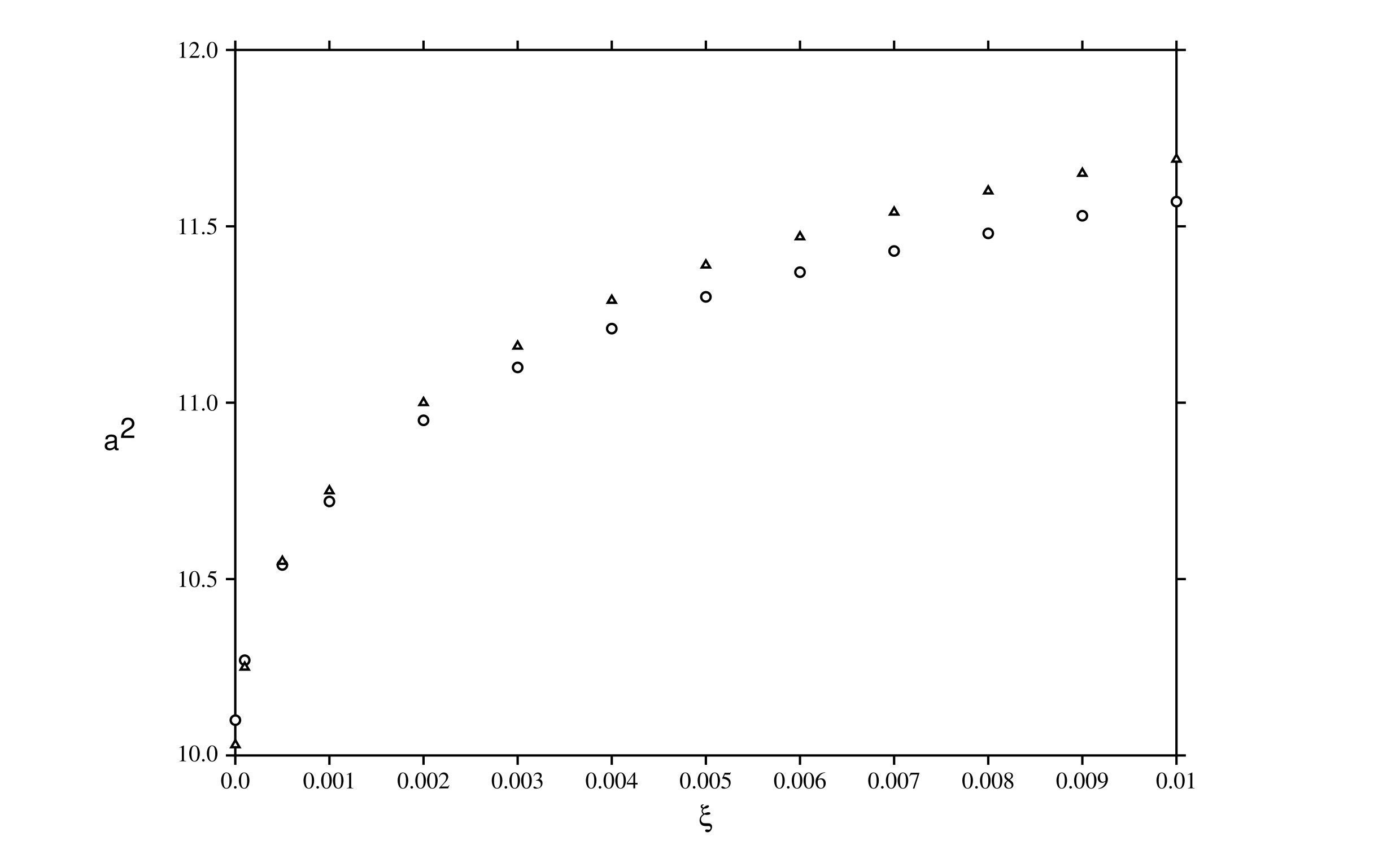}
			\caption{
				Graph of 
				$a^2$ 
				versus 
				$\xi$. 
				$\circ$ denotes $\kappa=0.01$; 
				$\triangle$ denotes $\kappa=10^{-3}$; 
				$\xi$ runs from $10^{-6}$ to 0.1.
			}
			\label{fig:fig4}
   \end{center}
	\end{figure}
	\newpage
	
	\section{Conclusions} 
	\label{S:Conclusions}
	
	We have investigated a model for thermal convection employing a generalized Navier-Stokes theory which includes bi-Laplacian terms
	of both the velocity and temperature fields.
	Such a hydrodynamic model 
	is physically very relevant in current research since, for example, \cite{GreenRivlin:1964,GreenRivlin:1967} and \cite{Green:1965} 
	argue that such extra spatial derivatives will be  
	important when the molecular structure of the fluid involves long molecules. In addition, such a model fits well in the rapidly
	expanding
	industry of microfluidics where length scales are very small, see \cite{FriedGurtin:2006}. We have incorporated higher gradients of temperature into the model and this fits in with similar research in viscoelasticity by
	\cite{Fabrizio:2022jts}. 
	It has been shown that the results of linear instability are practically the same as those found from a 
global nonlinear energy stability analysis. 
	This is very important and demonstrates that the key physics is incorporated by utilizing linear instability theory.
	
	Currently energy production is a vital topic affecting everyone. In this regard \cite{zahra2016new} describe a new 
	method which involves heating and cooling a ceramic
	plate positioned above a container of oil which is undergoing convective thermal motion. The variation of temperature
	in the ceramic plate produces electricity by means of the the pyroelectric effect. 
	It is interesting to ask whether a fluid with long molecules, or a suspension, 
	in a situation where micro-length scales dominate, would improve this technique of generating electricity. This could
	be a genuine use for the theory proposed here.
	
	Another relevant area in renewable energy is solar pond technology. Recent research is adding phase change and other materials to
	salt water to increase efficiency of the solar pond distillation and electricity production, cf. \cite{Mahfoudh:2019},
	\cite{Yu:2022}. Addition of such materials will change the molecular structure of the fluid and is likely to
	be suited to higher order velocity and temperature gradients.
	The work described herein is suitable for a description of a solar pond since it predicts a significantly 
	increased critical Rayleigh number. This means that the threshold before convective instability begins is larger and this is
	highly useful in a solar pond where one does not wish convective motion to ensue. 
	
	Thermal convection in nanofluids is very topical in heat transfer and renewable energy research, see e.g. \cite{ChangRuo:2022}.
	A nanofluid is typically a suspension of tiny particles of a metallic oxide in a carrier fluid and there is definite 
	evidence that a suspension does not behave 
	like a Navier-Stokes fluid, see e.g. \cite{KwakKim:2005}. A copper oxide nanofluid suspension contains 
	particles of the shape of a prolate spheroid of aspect ratio 3, see \cite{KwakKim:2005}. Such a molecular liquid is known to 
	display behaviour not commensurate with Navier-Stokes theory, see e.g. \cite{Travis:1997}, 
	where a flattened velocity profile is observed in Poiseuille flow
	instead of the parabolic one of classical fluid mechanics. 
	The higher order velocity and temperature gradient theory described here does not suffer from the drawback of a parabolic profile. 
	Hence, we believe the theory proposed here is suitable for the basis of
	a proper description of convection in a nanofluid
	suspension.
	
	To conclude we observe that stimulating recent work of \cite{Moon:2019}, \cite{Moon:2021EPJP} 
	has analysed interesting attractors and behaviour for
	ordinary differential equation systems derived from double diffusive convection using Navier-Stokes theory. 
	It is an interesting question to analyse how the inclusion of higher spatial gradients of both velocity and temperature
	would affect the attractor behaviours.\\

	\noindent{\bf Acknowledgments}. The authors would like to thank Marco Degiovanni for providing some helpful suggestions. The work of BS was supported by an Emeritus Fellowship of the Leverhulme Trust, EM-2019-022/9. GG, AG, CL, AM and AM are supported by Gruppo Nazionale per la Fisica Matematica (GNFM) of Istituto Nazionale di Alta Matematica (INdAM).
	\bibliographystyle{acm}

\end{document}